\documentclass[aps,prl,reprint,twocolumn,floatfix,hidelinks,superscriptaddress,nofootinbib]{revtex4-2}

\usepackage{braket}  % load before quantikz package
\usepackage{graphicx, xcolor}
\usepackage{siunitx}
\usepackage{amsmath, amssymb}
\usepackage{csquotes}
\usepackage{booktabs, multirow, bigdelim}
\usepackage{soul}
\usepackage[hidelinks]{hyperref}
\usepackage[nolist,smaller]{acronym}

\usepackage{amsthm, algorithm, algorithmic}
\floatstyle{ruled}
\newfloat{protocol}{H}{idf}
\floatname{protocol}{Protocol}
\newfloat{resource}{H}{idf}
\floatname{resource}{Resource}
\newfloat{simulator}{H}{idf}
\floatname{simulator}{Simulator}
\newtheorem{lemma}{Lemma}
\newcommand{\Title}{Verifiable blind quantum computing with trapped ions and single photons}

\newcommand{\ion}[2]{\mbox{$^{#2}$#1$^+$}}
\newcommand{\Ca}{\ion{Ca}{43}}
\newcommand{\Sr}{\ion{Sr}{88}}

\newcommand{\hfslevshort}[2]{F\!=\!{#1},\,m_{F}\!=\!{#2}}

\newcommand{\ish}{\mbox{$\sim$}\,}
\newcommand{\ltish}{\protect\raisebox{-0.4ex}{$\,\stackrel{<}{\scriptstyle\sim}\,$}}
\newcommand{\gtish}{\protect\raisebox{-0.4ex}{$\,\stackrel{>}{\scriptstyle\sim}\,$}}
\newcommand{\ketbra}[2]{\mathinner{|{#1}\rangle\langle{#2}|}}
\newcommand{\eqsym}[1]{\protect\raisebox{-0.4ex}{$\,\stackrel{\scriptsize #1}{=}\,$}}

\begin{document}

% Include your paper's title here
\title{\Title{}}
\author{P.~Drmota}
\author{D.~P.~Nadlinger}
\author{D.~Main}
\author{B.~C.~Nichol}
\author{E.~M.~Ainley}
\affiliation{Department of Physics, University of Oxford, Clarendon Laboratory, Parks Road, Oxford OX1 3PU, U.K.}

\author{D.~Leichtle}
\affiliation{Laboratoire d’Informatique de Paris 6, CNRS, Sorbonne Universit\'e, Paris 75005, France}

\author{A.~Mantri}
\affiliation{Joint Center for Quantum Information and Computer Science, University of Maryland, College Park, U.S.}

\author{E.~Kashefi}
\affiliation{School of Informatics, University of Edinburgh, Edinburgh EH8 9AB, United Kingdom.}
\affiliation{Laboratoire d’Informatique de Paris 6, CNRS, Sorbonne Universit\'e, Paris 75005, France}

\author{R.~Srinivas}
\author{G.~Araneda}
\author{C.~J.~Ballance}
\author{D.~M.~Lucas}
\affiliation{Department of Physics, University of Oxford, Clarendon Laboratory, Parks Road, Oxford OX1 3PU, U.K.}

\date{\today}

\begin{abstract}
We report the first hybrid matter-photon implementation of verifiable blind quantum computing.
We use a trapped-ion quantum server and a client-side photonic detection system networked via a fibre-optic quantum link.
The availability of memory qubits and deterministic entangling gates enables interactive protocols without post-selection -- key requirements for any scalable blind server, which previous realisations could not provide.
We quantify the privacy at \ltish$0.03$ leaked classical bits per qubit.
This experiment demonstrates a path to fully verified quantum computing in the cloud.
\end{abstract}

\maketitle

\begin{figure}[ht]
    \centering
    \includegraphics{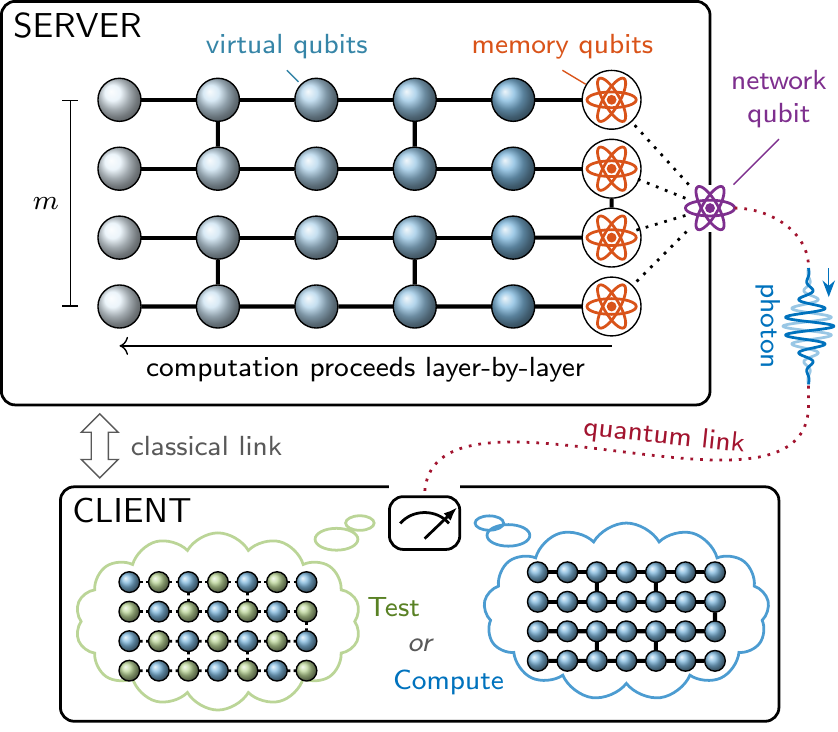}
    \caption{%
        Verifiable blind quantum computing in the measurement-based model.
        The computation is expressed as a sequence of measurements on a brickwork state (two-dimensional graph with vertices representing virtual qubits, and edges indicating $\mathsf{CZ}$ gates).
        The server holds $m$ physical memory qubits (orange atoms) and one physical network qubit (violet atom).
        The server can entangle these qubits deterministically with each other.
        The network qubit can also be entangled with a photon; by measuring this photon, the client can steer the network qubit in the server remotely without the server learning about its state.
        This allows the client to hide the computation (inputs, outputs, and circuit) from the server.
        Moreover, the client can verify that the computation has not been tampered with by (randomly) interleaving test rounds, which produce classically simulatable outcomes and cannot be distinguished from the actual computation by the server.
    }
    \label{fig:Intro}
\end{figure}

% Quantum computing introduction.
%
\noindent%
Quantum computers are poised to outperform the world's most powerful supercomputers, with applications ranging from drug discovery to cybersecurity.
These computers harness quantum phenomena such as entanglement and superposition to perform calculations that are believed to be intractable with classical computers.
As quantum processors control delicate quantum states, they are necessarily complex and physical access to high-performance systems is limited.
Cloud-based approaches, where users can remotely access quantum servers, are likely to be the working model in the near term and beyond;
many users already perform computations on commercially available devices for state-of-the-art research~\cite{sarma_quantum_2019, alcazar_classical_2020, proctor_measuring_2022, amaro_filtering_2022, kirsopp_quantum_2022}.

However, delegating quantum computations to a server carries the same privacy and security concerns that bedevil classical cloud computing.
Users are currently unable to hide their work from the server or to independently verify their results in the regime where classical simulations become intractable.
Remarkably, the same phenomena that enable quantum computing can leave the server \enquote{blind} in a way that conceals the client's input, output, and algorithm~\cite{broadbent_universal_2009, fitzsimons_unconditionally_2017, gheorghiu_verification_2019};
because quantum information cannot be copied and measurements irreversibly change the quantum state, information stored in these systems can be protected with information-theoretic security, and incorrect operation of the server or attempted attacks can be detected -- a surprising possibility which has no equivalent in classical computing.
\Ac{BQC} requires not only a universal quantum computer as the server, but also a quantum link connecting it to the client~\cite{badertscher_security_2020, cojocaru_possibility_2021}.
Photons are a natural choice to provide that link, and indeed the first demonstrations of \ac{BQC} were performed in purely photonic systems~\cite{barz_demonstration_2012, barz_experimental_2013, fisher_quantum_2014, greganti_demonstration_2016}.
However, unavoidable photon loss, either due to limited photon detection efficiencies or absorption in the link, results in potential security risks~\cite{barz_demonstration_2012, fisher_quantum_2014} and places hard limits on the scalability of this approach due to the resource overhead incurred by post-selection~\cite{li_resource_2015}.
Ideally, quantum information at the server should be stored in a stable quantum memory that can be manipulated with high fidelity, yet readily interfaced to a photonic link.
The ability to retain quantum information on the server then enables the client to perform adaptive mid-circuit adjustments in order to execute the target computation deterministically and securely.
Combining two completely different platforms at the single-quantum level is technically challenging~\cite{pfaff_unconditional_2014, hucul_modular_2015}; so far, quantum network nodes with integrated memory qubits have been realised with solid state systems~\cite{kalb_entanglement_2017, stas_robust_2022} and trapped atoms~\cite{wilk_single-atom_2007, drmota_robust_2023}.

% Our approach solves these issues.
Here, we demonstrate \ac{BQC} using a trapped-ion quantum processor (server) that integrates a robust memory qubit encoded in \Ca{} with a single-photon interface based on \Sr{} to establish a quantum link to the client (photon detection system).
We implement an interactive protocol, where the client can remotely prepare single-qubit states on the server adaptively from shot to shot using real-time classical feedforward control.
The complexity needed for universal quantum computation is contained entirely within the server, while the client is a simple photon polarisation measurement device that is independent of the size and complexity of the algorithm and supports near-perfect blindness by construction.
The client and the server are controlled by independent hardware and connected only by a classical signalling bus and an optical fibre.
Our system achieves noise levels below a certain threshold for which arbitrary improvements to the protocol security and success rate (robustness) are theoretically possible~\cite{leichtle_verifying_2021}.

\paragraph{Protocol.}
% Measurement-based framework.
Quantum algorithms can be described in the measurement-based quantum computing model, which prescribes a sequence of measurements on a highly entangled resource state~\cite{raussendorf_one-way_2001, nielsen_cluster-state_2006}.
% Making it blind.
Information-theoretic blindness can be achieved, even against maliciously operating servers, if either the state preparation or the measurements are performed by the client~\cite{childs_unified_2005, broadbent_universal_2009, morimae_blind_2013, fitzsimons_private_2017}.

% Adding verification.
In the presence of noise, even a faithfully operating server produces erroneous results that are indistinguishable from nefarious modifications to the honest protocol~\cite{aharonov_interactive_2017, fitzsimons_unconditionally_2017, broadbent_how_2018, gheorghiu_verification_2019}.
Blindness allows the client to secretly test the quantum resources provided by the server.
The protocol implemented here achieves this by interleaving \enquote{computation} and \enquote{test} rounds.
% Handling noisy devices.
A statistical argument provides bounds for the security and robustness of this protocol for the important class of \ac{BQP} decision problems~\cite{leichtle_verifying_2021}.
The client accepts a result if the observed fraction of failed test rounds, $p_\mathrm{fail}$, is below a chosen threshold, $\omega$.
If $\omega$ is below the theoretical threshold $\omega_\mathrm{max}$, the overhead due to repetition is low: the probability of accepting an incorrect result decreases exponentially with the number of rounds.
The minimum value for $\omega$ depends on the amount of noise in the devices.
The client assumes a maximum expected test round failure rate, $p_\mathrm{max}$, and chooses $\omega > p_\mathrm{max}$ such that the probability of rejecting any result also decreases exponentially with the number of rounds, making the protocol robust to a limited amount of noise.

For universal quantum computation, particular graph states and a discrete set of single-qubit measurements, $\{\hat{B}_\alpha = \cos(\alpha) \mathsf{X} + \sin(\alpha) \mathsf{Y}\}_{\alpha\in\Theta}$, are sufficient~\cite{mantri_universality_2017}, where $\Theta = \{0,\pi/4,\dots,7\pi/4\}$, and $\mathsf{X},\mathsf{Y}$ are Pauli operators.
Graph states are specific multi-qubit states in which vertices represent qubits initialised in $\ket{+} = (\ket{0}+\ket{1})/\sqrt{2}$ and edges represent entanglement created by two-qubit $\mathsf{CZ}$ gates [Fig.~\ref{fig:Intro}], where $\mathsf{CZ} = \ketbra{0}{0}\otimes\mathsf{1} + \ketbra{1}{1}\otimes\mathsf{Z}$.
The qubits are measured in a fixed order, using the basis $\hat{B}_{\alpha_\ell}$ at node $\ell$, where $\alpha_\ell$ depends on the algorithm and on previous measurement outcomes.

\begin{figure*}[htb]
    \centering
    \includegraphics{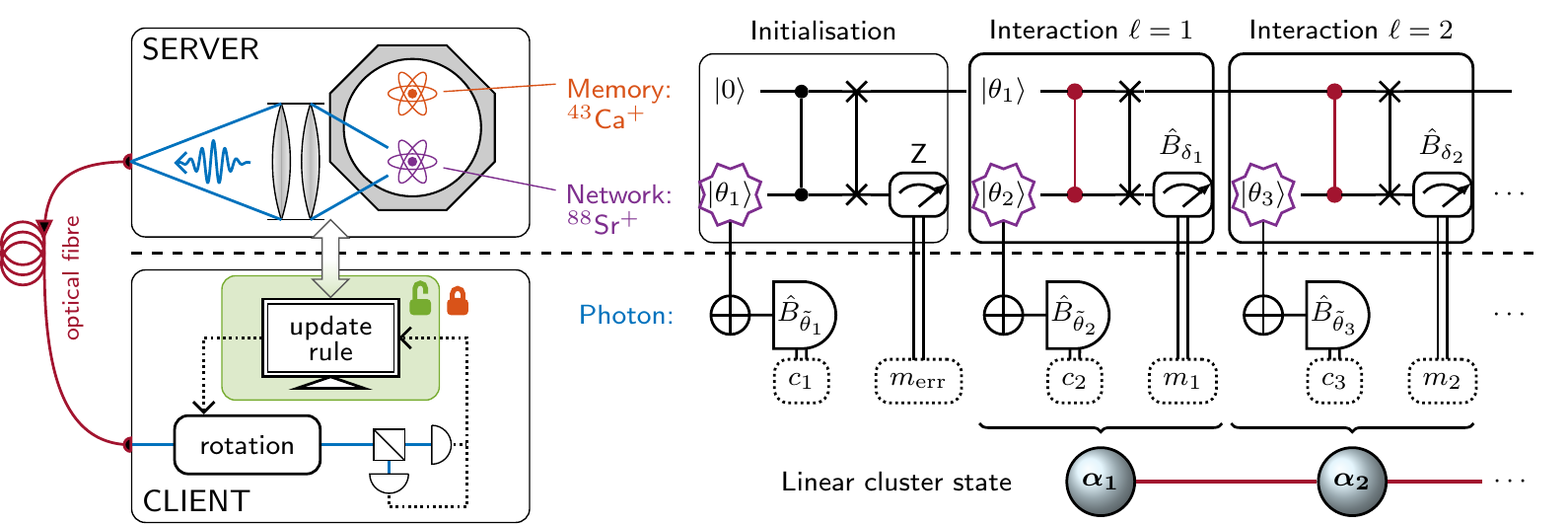}
    \caption{
        Protocol used to generate a linear cluster state using a trapped-ion quantum server and a photonic client.
        The client can steer the network qubit into $\ket{\theta_\ell} = \ket{\tilde{\theta}_{\ell}+c_{\ell} \pi}$ by measuring the polarisation of the photon in the basis $\hat{B}_{\tilde{\theta}_{\ell}}$ and obtaining $c_\ell \in \{0,1\}$ as outcome.
        In the initialisation step, the server transfers this state onto a memory qubit such that the network qubit can be steered again~\cite{drmota_robust_2023}.
        Every subsequent interaction step extends the size of the cluster state;
        the client steers the network qubit remotely into $\ket{\theta_{\ell+1}}$, the server entangles it ($\mathsf{CZ}$ gates), and performs a measurement in the basis $\hat{B}_{\delta_\ell}$, where $\delta_\ell$ is provided by the client.
        See text for details.
    }
    \label{fig:Protocol}
\end{figure*}
% More concrete instructions
To blindly run the above protocol with measurement angles $\alpha_\ell$, the client performs \ac{RSP} into $\ket{\theta_\ell} = \exp(-\mathrm{i} \frac{\theta_\ell}{2} \mathsf{Z}) \ket{+}$, with secret phase shift $\theta_\ell \in \Theta$ for every qubit $\ell=1,2,\dots,q$, and shifts the measurement angles accordingly.
This way, $\theta_{\ell}$ act as a classical encryption key such that $\alpha_\ell$ remain private to the client.
To ensure that the corresponding measurement outcomes, $m_\ell \in \{0,1\}$, are uninformative, the client hides bit flips in half of the measurement angles that are indicated by secret key bits, $r_\ell \in \{0,1\}$ [Eq.~\eqref{eq:Encryption}].
The client can recover the unencrypted measurement outcomes as $m_\ell \oplus r_\ell$.

% Limiting to linear cluster applications.
Here we implement \ac{BQC} on linear cluster states [Fig.~\ref{fig:Protocol}].
Two physical qubits are sufficient to implement linear clusters of arbitrary length, as qubits can be reinitialised after every mid-circuit measurement.
The first qubit -- the network qubit -- can be steered into an arbitrary state by the client using \ac{RSP}~\cite{bennett_remote_2001}, while the second qubit -- the memory qubit -- carries the information encoded in the leading node of the expanding linear cluster state.
We break the cluster state into discrete interaction steps between the server and the client, starting with the initialisation step [Fig.~\ref{fig:Protocol}], which prepares the memory qubit in $\ket{\theta_1}$.
At each interaction of a computation round, the client performs \ac{RSP} to steer the network qubit into $\ket{\theta_{\ell+1}}$ and communicates
\begin{align}
    \delta_\ell=(-1)^{R_{\ell - 1}} \alpha_\ell + \theta_\ell + \pi r_\ell
    \label{eq:Encryption}
\end{align}
to the server, where $R_\ell = \bigoplus_{1 \leq j < \ell/2} (m_{\ell - 2j} \oplus r_{\ell - 2j})$ is the adaptive feedforward correction from decrypted previous measurements.
After applying the $\mathsf{CZ}$ gate and a $\mathsf{SWAP}$ gate, the server measures the network qubit in the $\hat{B}_{\delta_\ell}$ basis and returns the result, $m_\ell$, to the client [interaction blocks in Fig.~\ref{fig:Protocol}].
This process leaves the leading cluster state node on the memory qubit, encrypted by $R_{\ell}$~\cite{supplement}, while the network qubit is available for further \ac{RSP}.
\nocite{tomescu_qubit_2019, ziegler_optimum_1942, simon_hamiltons_2012, rehacek_diluted_2007}

The client randomly assigns each round a secret label identifying them as a computation or a test.
In test rounds, the client prepares every second qubit in a $\mathsf{Z}$ eigenstate, $\ket{r_\ell}$, which are called \enquote{dummy qubits}.
This step leaves the remaining, so-called \enquote{trap qubits}, in a separable state.
The outcome $m_\ell \eqsym{!} r_\ell$ of measuring these trap qubits with $\delta_\ell = \theta_\ell + \pi r_\ell$ can thus be predicted efficiently by the client.

%\section*{Implementation}
\paragraph{Server.}
The server controls an ion trap quantum processor containing one \Sr{} and one \Ca{} ion.
Ion-photon entanglement needed for \ac{RSP} is generated by fast excitation and spontaneous decay \cite{blinov_observation_2004} on the \SI{422}{\nano\meter} transition of \Sr{}.
The single photons are collected by free-space optics and coupled into a single-mode optical fibre~\cite{stephenson_high-rate_2020}, which forms the quantum link with the client.
The memory qubit is encoded in \Ca{}, which provides a long coherence time (\ish\SI{10}{\second}) and is unaffected by concurrent manipulation of \Sr{}~\cite{drmota_robust_2023}.
Thus, \Sr{} can be used for mid-circuit measurements and sympathetic cooling between interaction steps.
The $\mathsf{CZ}$ gate required to build the cluster state is combined with the $\mathsf{SWAP}$ gate into an $\mathsf{iSWAP}$ gate.
This enables reuse of \Sr{} for \ac{RSP} whilst the current state of the computation is retained on the memory qubit.
Errors during the initialisation step are detected in real time [$m_\mathrm{err} = 1$ in Fig.~\ref{fig:Protocol}] in which case this step is repeated.
\begin{figure}[tb]
    \centering
    \includegraphics{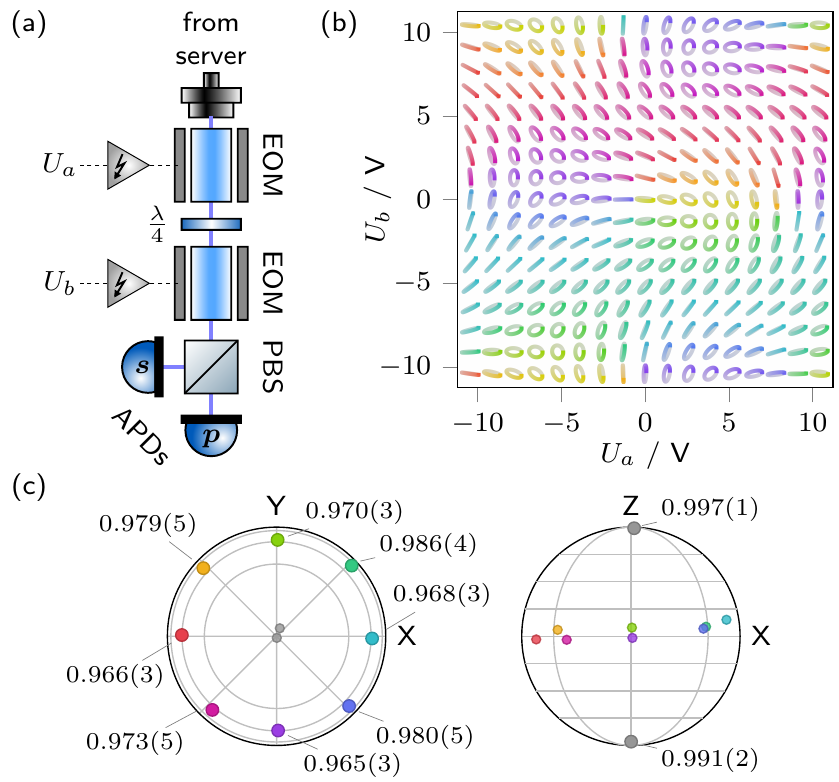}
    \caption{The client performs \acf{RSP} using a fast-switching polarisation analyser.
        (a) The control voltages ($U_a$, $U_b$) of two \acp{EOM} separated by a $\lambda/4$ waveplate enable the client to arbitrarily rotate the measurement basis given by the \acs{PBS}.
        (b) Laser light is used to reconstruct this basis for different $U_a, U_b$.
        Polarisation ellipses are shown for the basis states heralded by detector $p$, where the colour represents their phase.
        (c) To find $U_a, U_b$ which maximise the fidelity $F$ to each target state needed during the protocol, we perform tomography on the network qubit after \ac{RSP}.
        The averaged results from 36 calibrations over 2 weeks are shown in the Bloch sphere representation of the network qubit.
        Values indicate $F$, with standard deviations obtained from bootstrapping.
    } \label{fig:PolarisationAnalyser}
\end{figure}

\paragraph{Client.}
The client receives single photons from the server through an optical fibre.
The quantum capability of the client is reduced to projective polarisation measurements of these photons in a basis that can be dynamically reconfigured by changing the voltages on two \acfp{EOM}~\cite{supplement} [Fig.~\ref{fig:PolarisationAnalyser}].
This measurement remotely steers the network qubit into a state that depends only on the polarisation measurement basis and the measurement outcome obtained, information known exclusively to the client [$\tilde{\theta}_\ell$ and $c_\ell$ in Fig.~\ref{fig:Protocol}].
Birefringence in the optical fibre transforms the photonic state before reaching the client by an unknown unitary operation, which drifts on a timescale of \ish\SI{10}{\minute} due to thermal effects.
To compensate for this drift, the client periodically recalibrates the \ac{EOM} voltages~\cite{supplement} [Fig.~\ref{fig:PolarisationAnalyser}(c)].

\paragraph{Blindness.}
We consider information that could leak to an adversarial server, concerning the client's polarisation measurement, via the network qubit, which is controlled by the server, and through classical signals, which are controlled by the client.
We quantify the information that the server could gain from measuring the network qubit at $0.031(4)$ bits per interaction step using quantum state tomography, and find good agreement with independent estimates~\cite{supplement}.
In our demonstration, mismatched electronic delays between different polarisation measurement outcomes are the dominant cause for information leakage.
However, as the client controls the relevant classical signals, these delays could be matched.
The remaining leakage of \ish\num{0.001} bits per interaction step would be dominated by imperfections in the polarising optics used by the client.

\begin{figure}[htb]
    \centering
    \includegraphics{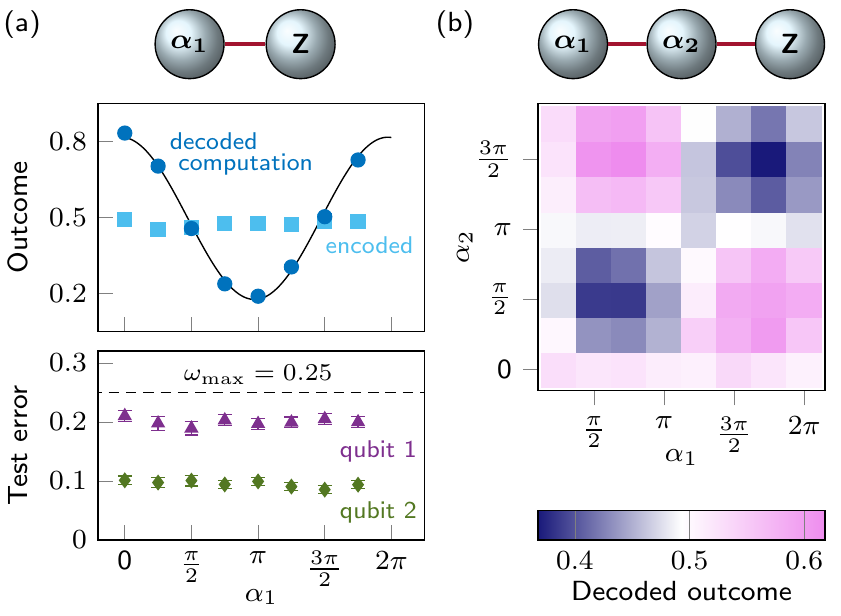}
    \caption{
        Experimental results on an expanding linear cluster state, where the leading qubit is measured in the $\mathsf{Z}$ basis after (a) one, and (b) two interaction steps between the client and the server.
        (a) While the server observes mixed outcomes (squares, \ish\num{2000} test and computation rounds each), for each $\alpha_1$, the client can decode the results using the secret keys.
        A fit to the decoded computation outcomes (circles) is shown to guide the eye.
        Error intervals indicate the binomial standard error.
        The test round errors are significantly below the threshold for verification of a two-node cluster state (dashed line).
        (b) The decoded outcome is shown for different blind measurement settings, ($\alpha_1, \alpha_2$), each comprising \ish\num{3100} computation rounds (see Supplementary Material~\cite{supplement} for interleaved test round results).
    }
    \label{fig:Results}
\end{figure}
\paragraph{Results.}
We realise different quantum computations with one and two interaction steps, see Figs.~\ref{fig:Results}(a) and \ref{fig:Results}(b) respectively.
We could use the output qubit in further interaction steps, or make a final measurement in the basis $\hat{B}_{\delta_{q+1}}$ to complete the $(q+1)$-node cluster computation.
In this demonstration, however, the output qubit is always measured in the $\mathsf{Z}$ basis.
Since this measurement commutes with the $\mathsf{CZ}$ gate preceding it, the computation is equivalent to a cluster state with one fewer node.
The one- and two-step interactions therefore implement the computations $\mathsf{H} Z(\alpha_1) \ket{+}$ and $X(\alpha_2) Z(\alpha_1) \ket{+}$, respectively, where $\mathsf{H}$ is the Hadamard gate, $X(\alpha) = \exp(-\mathrm{i} \frac{\alpha}{2} \mathsf{X})$ and $Z(\alpha)=\exp(-\mathrm{i} \frac{\alpha}{2} \mathsf{Z})$ are single-qubit rotations, and $\alpha_1$ and $\alpha_2$ are encrypted using Eq.~\eqref{eq:Encryption} during the protocol.
From the server's perspective, the outcomes appear random [squares in Fig.~\ref{fig:Results}(a)] as a result of the bit-flip encryption, $\delta_\ell \propto r_\ell \pi$, which is applied by the client in both the computation and test rounds.
The client on the other hand can use the round type (computation or test) and encryption key ($r_\ell$) to decode the outcomes.
The decoded computation outcomes, indicated by the circles in Fig.~\ref{fig:Results}(a) and the colourmap in Fig.~\ref{fig:Results}(b), match the expected fringe pattern as a function of the blind measurement angles $\alpha_1$ and $\alpha_2$.
Experimental imperfections lead to a reduction in contrast and to phase shifts.
The client observes an error rate of $p_\mathrm{fail}^{(1)} = 0.201(3)$ on the first qubit and $p_\mathrm{fail}^{(2)} = 0.095(2)$ on the second qubit [bottom panel in Fig.~\ref{fig:Results}(a)], which are consistent with known error sources~\cite{supplement}.
By changing the final measurement basis from $\mathsf{Z}$ to $\hat{B}_{\delta_{q+1}}$ with an additional $\pi/2$ pulse, which would have no significant impact on the error budget, and randomly choosing one qubit as trap qubit in every test round, we find that a two-node cluster computation could be verified using our apparatus~\cite{supplement};
the expected average test round failure probability of \ish\num{0.18} would be significantly below $\omega_\mathrm{max} = 0.25$ required for secure and robust verification of this state.
The corresponding test round results for the three-node cluster computation are shown in the Supplementary Material~\cite{supplement};
the observed failure rates indicate that verification is not possible in this case, largely due to technical limitations (motional heating) on the $\approx 0.91$ fidelity of the $\mathsf{iSWAP}$ gate~\cite{drmota_robust_2023}.

\paragraph{Conclusion.}
We have implemented a protocol for blindly delegating quantum computations to a trapped-ion quantum processor, using a client apparatus that requires only single-photon polarisation measurements and classical communication.
We have established bounds on information leakage through both the classical and quantum channels that are present in our implementation.
We have shown that the size of the cluster state can be increased without increasing the number of physical qubits in the server and without modifications to the client hardware.
If more memory qubits were added to the server~\cite{wright_benchmarking_2019, keller_controlling_2019}, the computational space could be extended to higher-dimensional cluster states.
We have taken steps to include verification into the protocol, and the measured test round error indicates that computations on two-node cluster states could be verified robustly and reliably.
We predict that for a \ac{BQP} decision problem with small inherent algorithmic error and $p_\mathrm{max}=0.185$, the probability of accepting an incorrect result and that of rejecting any result would both be $\num{e-5}$ after \num{24000} repetitions, including \num{14400} test rounds; every additional \num{1200} repetitions would halve this likelihood~\cite{supplement}.
This approach is expected to provide both security and robustness for larger cluster states and other algorithms as long as the errors remain below the size-dependent threshold, $\omega_\mathrm{max} \approx 1 - (3/4)^{2/q}$, where $q$ is the total number of qubits in the cluster state.
The protocol that we have implemented does not incorporate error correction; to verify larger cluster states, the error per interaction step would need to be reduced.
The infidelity of the $\mathsf{iSWAP}$ gate is the leading error source~\cite{drmota_robust_2023}, but we note that in other systems, $\mathsf{CZ}$ gates between \Sr{} and \Ca{} with fidelity $0.998$ have been demonstrated~\cite{hughes_benchmarking_2020}.
The state-of-the-art ion-photon entanglement fidelity of $0.979(1)$ (this apparatus) is limited primarily by technical imperfections in the optical setup (alignment).

In comparison with previous experimental implementations~\cite{barz_demonstration_2012, barz_experimental_2013, fisher_quantum_2014, greganti_demonstration_2016}, which were based on purely photonic platforms without quantum memory, this work overcomes several major challenges associated with real-world \ac{BQC} deployments.
As quantum logic operations in the server are deterministic and the interaction with the client is heralded, our implementation eliminates the need for post-selection, avoiding the associated efficiency, scalability, and security issues~\cite{barz_demonstration_2012, barz_experimental_2013, greganti_demonstration_2016}.
Here, photon losses in particular do not present a security threat, and the use of a memory qubit combined with fast and adaptive hardware facilitates true shot-by-shot randomisation of all protocol parameters in real time.

Future realisations could involve a complex network of servers and clients.
Photons could be routed to a number of clients using optical switches, and the distance to the server increased using frequency conversion of the photons to telecommunication wavelengths~\cite{krutyanskiy_light-matter_2019} or using recent developments in fibre technology~\cite{numkam_loss_2023}.
The photonically-interfaced trapped-ion quantum information platform demonstrated here paves the way for secure delegation of confidential quantum computations from a client with minimal quantum resources to a fully capable, but untrusted, quantum server.\\

% Acknowledgements.
We thank Sandia National Laboratories for supplying the HOA2 ion trap used in this experiment, and the developers of the experimental control system ARTIQ~\cite{ARTIQ}.
DPN acknowledges support from Merton College, Oxford.
DL acknowledges support from the ANR project SecNISQ.
AM and DM acknowledge support from the U.S.\ Army Research Office (refs.\ W911NF-20-1-0015 and W911NF-18-1-0340).
GA consults for Nu Quantum Ltd and acknowledges support from Wolfson College, Oxford.
RS is partially employed by Oxford Ionics Ltd and acknowledges funding from an EPSRC Fellowship EP/W028026/1 and Balliol College, Oxford.
CJB is a director of Oxford Ionics and acknowledges support from a UKRI FL Fellowship.
EK acknowledges support from grant ref.\ EP/X026167/1.
We acknowledge technical and financial support from VeriQloud (of which EK is a co-founder) during the initial design of this project, via an industry partnership grant from the NQIT Quantum Technology Hub EP/M013243/1.
This work was supported by the U.K.\ EPSRC \enquote{Quantum Computing and Simulation} Hub EP/T001062/1, and the E.U.\ Quantum Technology Flagship Project AQTION (No.\ 820495).

\onecolumngrid
\clearpage
\begin{center}
    \textbf{\large Supplemental Material for `\Title{}'}
\end{center}
\twocolumngrid

\renewcommand{\thepage}{S\arabic{page}}
\renewcommand{\thesection}{S\arabic{section}}
\renewcommand{\thetable}{S\arabic{table}}
\renewcommand{\thefigure}{S\arabic{figure}}

\section{Data handling}
We make extensive efforts to conduct our experiments under conditions that are representative of a real deployment.
The client and the server are controlled by independent personal computers and hardware from the ARTIQ open-source control system~\cite{ARTIQ}.
The experiment and calibration schedule is coordinated over Ethernet.
For timing-critical communication, such as the interaction during the protocol, low-latency electronic signals are used.
Throughout the data acquisition and analysis process, the client does not reveal any protocol secrets to the server.

\section{Linear cluster state}
In computation rounds, the memory qubit state after $q$ interactions is given by
\begin{gather*}
    \ket{\psi_{q+1}} = \hat{Z}(\theta_{q+1}) \mathsf{X}^{R_{q}} \mathsf{Z}^{R_{q-1}} \left(\mathsf{H}\hat{Z}(\alpha_{q}) \cdots \mathsf{H}\hat{Z}(\alpha_1)\right) \ket{+}, \\
    \hat{Z}(\alpha) := \exp\left( - \mathrm{i} \frac{\alpha}{2} \mathsf{Z} \right), % XXX: Check convention.
\end{gather*}
where $\mathsf{X}$ and $\mathsf{Z}$ are Pauli operators and $\mathsf{H}$ is the Hadamard gate.

\section{Sequence timings}
The time taken to process one node of a cluster state includes an average \SI{100}{\micro\second} until successful detection of a single photon at the client for \ac{RSP} (limited by photon loss in the quantum link), \ish\SI{400}{\micro\second} for transfer between the logic and the memory qubit in \Ca{}~\cite{drmota_robust_2023}, \ish\SI{340}{\micro\second} for the $\mathsf{iSWAP}$ gate between \Sr{} and \Ca{}, \ish\SI{135}{\micro\second} for readout of \Sr{}, \SI{50}{\micro\second} for deshelving of \Sr{}, \ish\SI{230}{\micro\second} for sympathetic ground state cooling using \Sr{}, and \ish\SI{150}{\micro\second} for the communication of the measurement outcome from the server to the client.

The duration of each photon generation attempt is \SI{1}{\micro\second}, which includes server-side fast state preparation of \Sr{} (\ish\SI{350}{\nano\second} laser switching latency $+$ \SI{300}{\nano\second} pulse duration $+$ \SI{50}{\nano\second} delay), server-side pulsed excitation within a \SI{12.5}{\nano\second} window, client-side photon detection windows (\ish\SI{30}{\nano\second}), and communication of the outcome, i.e.\ whether a photon was received (\ish\SI{64}{\nano\second}).
The client and the server continue attempts in a loop until either a photon is detected at the client or a timeout (\SI{1}{\milli\second}) is reached.
In the case of a timeout (probability $<\num{e-4}$), a series of system checks (ion loss, laser lock status) is performed before continuing the protocol.

\section{Client apparatus}
In order to meet the timing requirements for rapid manipulation of the photon polarisation, \acp{EOM} (Thorlabs, EO-AM-NR-C4) were selected for their fast switching speeds.
For achieving universality, two electro-optic modulators in series provide the necessary degrees of freedom to be able to realise any polarisation measurement basis.

\begin{figure}[htb]
    \centering
    \includegraphics{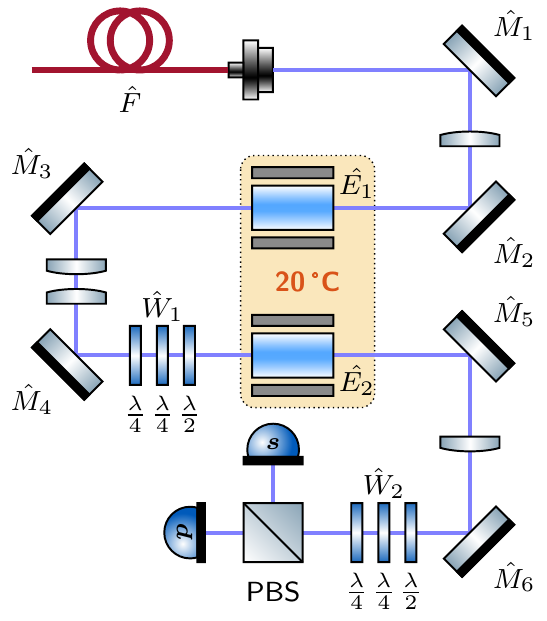}
    \caption{Detailed beam path for the fast photon polarisation analyser used by the client.
        The \acp{EOM} are thermally shielded and actively stabilised to \SI{20}{\degreeCelsius} with \SI{0.1}{\milli\kelvin} stability.
        The waveplate triplets $\hat{W}_1$ and $\hat{W}_2$ are used to cancel unwanted birefringence in the sections comprising $\{\hat{M}_3, \hat{M}_4\}$ and $\{\hat{E_2}, \hat{M}_5, \hat{M}_6\}$, respectively.
        The birefringence in components $\{\hat{F}, \hat{M}_1, \hat{M}_2, \hat{E_1}\}$ is absorbed into the calibration of the device.
    }
    \label{fig:client apparatus}
\end{figure}

\subsection{Preliminaries}

\subsubsection{Temperature stability of EOMs}
Even though the \acp{EOM} are manufactured in a dual-crystal configuration which provides passive cancellation of the differential temperature dependence between the ordinary and extraordinary axes, significant temperature-dependent polarisation drifts were observed in a preliminary investigation~\cite{tomescu_qubit_2019}.
Therefore, the final optical layout of the client apparatus is designed with shared temperature stabilisation of the \acp{EOM} [Fig.~\ref{fig:client apparatus}].
A two-layer enclosures with added thermal insulation surrounds the \ac{EOM} modules.
In addition, active stabilisation is employed using a Peltier element on the top surface of the mount, with a heatsink and a fan for ducted heat removal.
With the feedback gains calibrated using a variant of the Ziegler-Nichols method~\cite{ziegler_optimum_1942}, the temperature settles within \ish\SI{2}{\minute} and reaches a stability below \SI{0.1}{\milli\kelvin} under typical operating conditions.

\subsubsection{Optical impurities in EOMs}
The incoming light is focussed into the first \ac{EOM}, then recollimated and focussed into the second \ac{EOM}, using \SI{300}{\milli\meter} plano-convex lenses.
Focussing the light through the crystals significantly reduces depolarising effects due to spatial inhomogeneities in the \ac{EOM} crystals.

\subsubsection{Polarising beam splitter}
Imperfections of the \ac{PBS} reduce the distinguishability of orthogonal polarisation states.
We measure the extinction ratio of the \ac{PBS} for pure $s$ and $p$ polarisation, and obtain $T_s/T_p=\num{0.5e-4}$ and $R_p/R_s=\num{1.3e-4}$ in transmitted and reflected power, respectively.
We note that this imperfection has no effect on the blindness of the implementation; it merely reduces the \ac{RSP} fidelity.

\subsubsection{Internal birefringence cancellation}
The action of an ideal \ac{EOM} is to rotate the polarisation around a fixed axis represented by the unitary transformation
\begin{align*}
    \hat{R}(U) = \exp\left[-\mathrm{i} \frac{\phi(U)}{2} \hat{X} \right] \ ,
\end{align*}
where the rotation angle $\phi$ is a function of the voltage $U$ applied across the crystal and $\hat{X}=\ketbra{H}{V}+\ketbra{V}{H}$ in the basis given by the extraordinary axis of the \ac{EOM}.
For a pair of ideal \acp{EOM} to be able to reach any output polarisation from an arbitrary input, the rotation axes must be made orthogonal to each other.
This can in principle be achieved by placing a quarter waveplate between the \acp{EOM}; in practice, however, there is an unknown amount of static birefringence from each optical element including the \ac{EOM} crystals and mirrors.
To compensate this exactly, the inverse unitary operation has to be constructed with optical elements.
It can be shown that at least a triplet of waveplates, e.g.\ two quarter-wave and one half-wave, are required to implement the most general unitary acting on the polarisation qubit~\cite{simon_hamiltons_2012}.
In Fig.~\ref{fig:client apparatus}, all optical components that could affect the polarisation are labelled with a unitary operator.
The operators $\hat{E_1}$ and $\hat{E_2}$ capture the unknown static birefringence in the two \acp{EOM}.
We adjust the waveplate triplets $\hat{W}_1$ and $\hat{W}_2$ to approximately cancel all unwanted sources of birefringence, such that
\begin{multline*}
    \hat{P}_\mathrm{BS} \underbrace{\hat{W}_2 \hat{M}_6 \hat{M}_5 \hat{E_2}}_{\mathsf{1}} \hat{R}(U_b)
    \underbrace{\hat{W}_1 \hat{M}_4 \hat{M}_3}_{\hat{Q}_\mathrm{WP}} \hat{R}(U_a)
    \underbrace{\hat{E_1} \hat{M}_2 \hat{M}_1 \hat{F} \ket{\psi}}_{\ket{\tilde{\psi}}} \\=
    \hat{P}_\mathrm{BS} \hat{R}(U_b) \hat{Q}_\mathrm{WP} \hat{R}(U_a)\ket{\tilde{\psi}} \ ,
\end{multline*}
where $\hat{P}_\mathrm{BS}=\hat{Z} = \ketbra{H}{H}-\ketbra{V}{V}$ is the projector implemented by the \ac{PBS} and $\hat{Q}_\mathrm{WP}=\exp(\mathrm{i} \frac{\pi}{4} \hat{Z})$ is the unitary of an ideal quarter waveplate aligned with a principal axis.
We do not correct the transformation from the input state $\ket{\psi}$ to $\ket{\tilde{\psi}}$ because this merely rotates the overall coordinate system.

\subsubsection{Switching dynamics}
Significant drift behaviour was observed after changing the electric field across the crystal~\cite{tomescu_qubit_2019}.
This can be attributed to charging of the crystal by the high-voltage electrodes attached to it.
We therefore apply compensation pulses after each pulse with the same duration and amplitude, but opposite sign, in order to zero the time-averaged electric field.
A systematic analysis of the pulse duration, settling time and the duty cycle showed that with compensation pulses in place the detrimental effects, which otherwise dominate, can be fully removed.
We determine the switching speed by recording the intensity of the transmitted fraction over time.
The intensity settles to \SI{1}{\percent} of the final value within \SI{18}{\micro\second}, in synchronisation with the settling of the driving voltage, which we therefore identify as the speed-limiting factor.

\subsection{Precharacterisation}
To characterise the action of the client setup on arbitrary polarisation inputs, we use motorised waveplates following a Glan-Taylor polariser at the input.
For this characterisation, we use \ish\SI{1}{\milli\watt} of continuous-wave laser light at \SI{422}{\nano\meter}.
The fraction of horizontally polarised output power contains information about the state amplitudes, but not their phase.
To gain knowledge of the phase, amplitude information of linearly independent input states must be combined.
We collect data for all combinations of quarter-waveplate (measured retardance $2 \pi \times \SI{0.2584(4)}{\radian}$) angles $q \in\{-\pi/4, 0, \pi/4\}\,\si{\radian}$ and half-waveplate (measured retardance $2 \pi \times \SI{0.5000(1)}{\radian}$) angles $h \in\{-\pi/8, 0, \pi/8, \pi/4\}\,\si{\radian}$ in random order.
For each \ac{EOM} voltage setting, we perform a nonlinear fit to the waveplate scan data using the model
\begin{align}
    t = \left|\bra{\psi} \hat{Q}_\mathrm{WP}(q) \hat{H}_\mathrm{WP}(h) \ket{H}\right|^2
    \label{eq:precharacterisation}
\end{align}
for the transmitted fraction $t$, where $\ket{\psi} = \cos(\vartheta/2) \ket{H} + \sin(\vartheta/2) \exp(\mathrm{i} \varphi) \ket{V}$, the waveplate angles $(q,h)$ are varied in the scan, and $(\vartheta, \varphi)$ are adjusted during the optimisation.
The result of this analysis is shown in Fig.~\ref{fig:PolarisationAnalyser}(b) of the main text.
Using a density matrix formulation of Eq.~\eqref{eq:precharacterisation}, the purity of the reconstructed states is found to be consistent with $\mathrm{Tr}(\rho^2) = 1$ for all \ac{EOM} voltage settings used.

\subsubsection{Orthogonal measurements}
In order to achieve perfect blindness, the measurements implemented by the client apparatus must not leak information to the server.
One possibility for this to happen would be secret-dependent noise, such as imperfections that depend on the \ac{EOM} voltage settings.
Using the same measurement setup as described at the beginning of this section, we record data for an exhaustive list of quarter- and half-waveplate angles (overlap between adjacent polarisation states $\approx\num{3e-4}$).
For each \ac{EOM} voltage pair $(U_a,U_b)$ scanned, we select the two input polarisation states $\ket{\psi_{ab}^{+}}$ and $\ket{\psi_{ab}^{-}}$ created by these waveplates that respectively maximise the power transmitted and reflected by the \ac{PBS}.
The overlap $\left|\braket{\psi_{ab}^{-}|\psi_{ab}^{+}}\right|^2$ shown in Fig.~\ref{fig:Orthogonality} averages to $0.0016$ (median value).
\begin{figure}[htbp]
    \centering
    \includegraphics{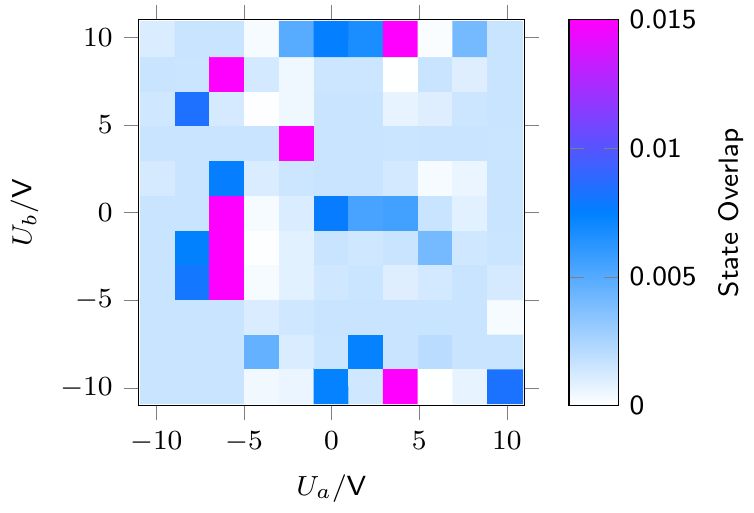}
    \caption{Overlap of polarisation input states which produce orthogonal polarisation measurement outcomes.}
    \label{fig:Orthogonality}
\end{figure}

\subsection{Long-term stability}
\begin{figure}[ht]
    \centering
    \includegraphics{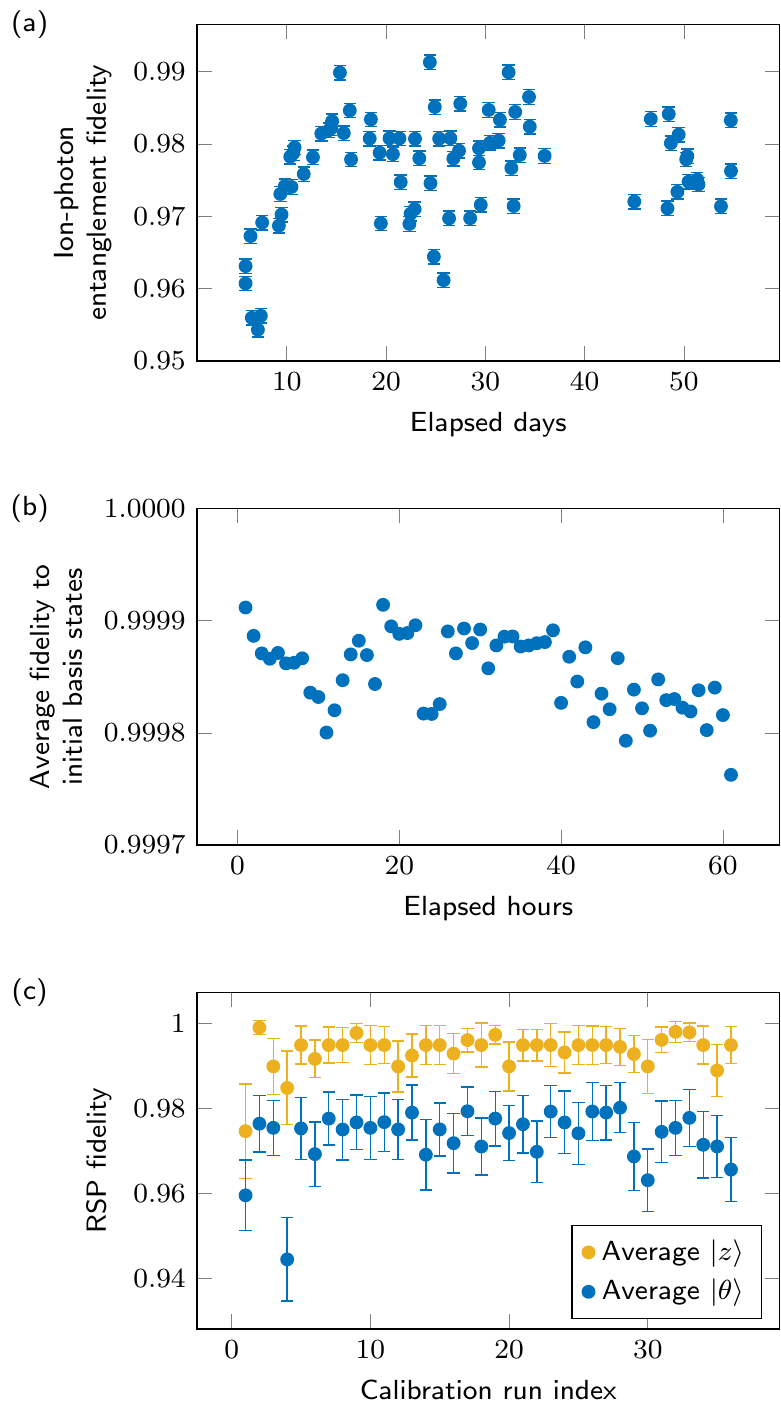}
    \caption{Long-term stability of (a) the photonic quantum networking link (b) the client apparatus and (c) the combined system.
        The measurements (a), (b), and (c) were performed at different times.}
    \label{fig:stability}
\end{figure}
\subsubsection{Ion-photon entanglement}
The performance of the ion-photon interface was monitored continuously using two-qubit tomography, which allows to reconstruct the density matrix that describes the joint state of the network qubit and the polarisation qubit [Fig.~\ref{fig:stability}(a)], as described in Ref.~\cite{stephenson_high-rate_2020}.

\subsubsection{Polarisation analyser}
The stability of the polarisation analyser was independently examined {\em ex situ} using repeated tomography measurements over the same set of \ac{EOM} control voltages for \SI{62}{\hour} [Fig.~\ref{fig:stability}(b)].
Over this time period, the infidelity due to polarisation drifts was less than \num{3e-4}, limited by the accuracy of the measurement.

\subsubsection{Remote state preparation fidelity}

We reconstruct the \enquote{steered} state of the network qubit using maximum likelihood tomography~\cite{rehacek_diluted_2007} for both polarisation heralds and average the fidelity to each of the target states needed for the verifiable blind quantum computing protocol.
The average fidelity of steering superposition states $\ket{\theta} = (\ket{0} + \exp(\mathrm{i}\theta)\ket{1})/\sqrt{2}$, where $\theta \in \{0, \frac{\pi}{4}, \frac{2\pi}{4}, \dots, \frac{7\pi}{4}\}$, and of steering $\mathsf{Z}$ basis eigenstates, $\ket{z}$, where $z \in \{0, 1\}$, is $\mathcal{F}_\theta = 0.973(7)$ and $\mathcal{F}_z = 0.996(3)$, respectively [Fig.~\ref{fig:stability}(c)].

\subsection{Remote state preparation calibration}
As the fibre connecting the server and the client is naturally exposed to changes in temperature and strain, its birefringence needs to be calibrated periodically.
In order to do so, the client instructs the server to perform $\mathsf{X}$-, $\mathsf{Y}$- and $\mathsf{Z}$-basis measurements for \ac{EOM} voltages on a regular $21\times 21$ grid, in random overall order.
Ion readout results are inverted for heralds in \ac{APD} $s$ and combined with results for \ac{APD} $p$.
Let $\Sigma_{\ket{y}}^{\mathsf{B}}$ denote the number of readout observations with outcome $y \in \{s,p\}$ when measured in the $\mathsf{B} \in \{\mathsf{X}, \mathsf{Y}, \mathsf{Z}\}$ basis.
Direct inversion tomography is used to reconstruct the Bloch vector $\vec{b} = (b_\mathsf{X}, b_\mathsf{Y}, b_\mathsf{Z})$ with components
\begin{align*}
    b_\mathsf{B} =
    \frac{\Sigma_{\ket{1}}^{\mathsf{B}} - \Sigma_{\ket{0}}^{\mathsf{B}}}
    {\Sigma_{\ket{1}}^{\mathsf{B}} + \Sigma_{\ket{0}}^{\mathsf{B}}} \ .
\end{align*}
The client then determines the \ac{EOM} voltage settings which minimise the angular deviation of $\vec{b}$ from the axis on the Bloch sphere corresponding to each of the 5 target bases required during the verifiable blind quantum computing protocol.

Combining the readout results from both \acp{APD} $s$ and $p$ leads to a reduction in \ac{RSP} fidelity if the states created by different heralds are not exactly orthogonal, for example due to a systematic delay mismatch in the heralding signal chains [Fig.~\ref{fig:blindness}(a)].
At the time the experiments were performed, this resulted in \ish\SI{5}{\percent} infidelity in each qubit in the cluster state.
This mismatch can be eliminated by matching the delays of the heralding signals, or by analysing the data separately for the two different heralds, such as for the fidelities shown in Fig.~\ref{fig:stability}(c).

\section{Trap failure probability for 2~interaction steps}
The trap qubit error rate is affected by remote steering into $\ket{\theta}$ (fidelity $\mathcal{F}_\theta' \approx 0.924$ lower than $\mathcal{F}_\theta$ due to unaccounted timing mismatch between \acp{APD} [Fig.~\ref{fig:blindness}(a)] when the data presented in the main text was gathered), the error-detected $\mathsf{iSWAP}$ gate (fidelity $\mathcal{F}_{\mathsf{iS}}' \approx 0.973$), mapping of the superposition between the $\{\ket{\hfslevshort{4}{4}}, \ket{\hfslevshort{3}{3}}\}$ logic qubit and the magnetic field-insensitive $\{\ket{\hfslevshort{4}{0}},\ket{\hfslevshort{3}{0}}\}$ qubit within the \Ca{} ground state hyperfine structure (fidelity $\mathcal{F}_\mathrm{map} \approx 0.98$), and the $\mathsf{iSWAP}$ without error detection (fidelity $\mathcal{F}_{\mathsf{iS}} \approx 0.913$).
We neglect state preparation, measurement, and single-qubit rotation errors (fidelities \gtish\num{0.99}).
The trap qubit failure rate is therefore expected to be
\begin{align*}
    p_\mathrm{fail}^{\mathrm{trap}_1} = 1 - \mathcal{F}_\mathrm{map} \mathcal{F}_{\mathsf{iS}} \mathcal{F}_\mathrm{map} \mathcal{F}_{\mathsf{iS}}' \mathcal{F}_\theta' \approx 0.21 \ .
\end{align*}
The dummy qubit error rate is affected by remote steering into $\ket{z}$ (fidelity $\mathcal{F}_z \approx 0.996$) and the final $\mathsf{iSWAP}$ gate:
\begin{align*}
    p_\mathrm{fail}^\mathrm{dummy} = 1 - \mathcal{F}_{\mathsf{iS}} \mathcal{F}_z \approx 0.09 \ .
\end{align*}
The measured values $p_\mathrm{fail}^{(1)}$ and $p_\mathrm{fail}^{(2)}$ that are stated in the main text agree well with these estimates.
We note that the last qubit in a cluster state must not be swapped back from the memory qubit to the network qubit; the number of $\mathsf{iSWAP}$ gates impacting this qubit is therefore reduced by one.
For the two-qubit cluster state, which could be implemented using the same sequence by changing the final measurement basis from $\mathsf{Z}$ to $\hat{B}_{\theta_{q+1}}$, we need to consider a second trap qubit placement: with the first qubit the dummy qubit and the second qubit the trap qubit, we expect
\begin{align*}
    p_\mathrm{fail}^{\mathrm{trap}_2} = 1 - \mathcal{F}_{\mathsf{iS}} \mathcal{F}_\theta' \approx 0.16 \ ,
\end{align*}
and the expected average test round failure rate is $p_\mathrm{fail} = \frac12 (p_\mathrm{fail}^{\mathrm{trap}_1} + p_\mathrm{fail}^{\mathrm{trap}_2}) \approx 0.18$.
If the \ac{APD} timing mismatch is taken into account in the \ac{EOM} voltage calibration (using $\mathcal{F}_\theta$ instead of $\mathcal{F}_\theta'$), the expected average test round failure rate reduces $\approx 0.14$.

\section{Test round results for 3 interaction steps}
\begin{figure}[htb]
    \centering
    \includegraphics{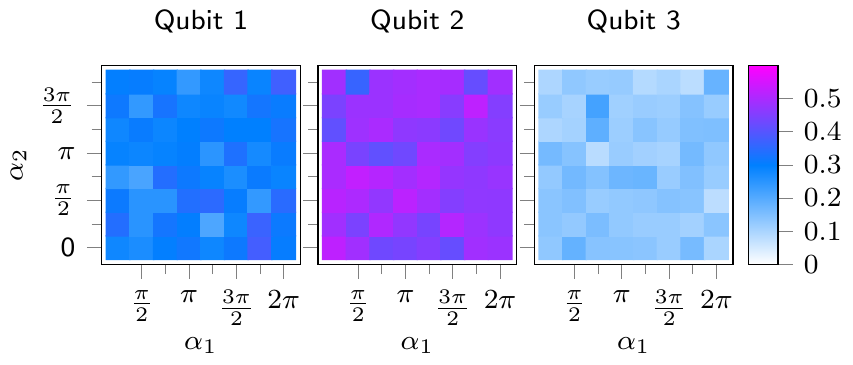}
    \caption{Observed trap failure rates on each of the qubits in the three-node cluster state (with qubit 3 measured in the $\mathsf{Z}$ basis).}
    \label{fig:rotation traps}
\end{figure}
We plot the error observed on each of the qubits in the three-node cluster state in Fig.~\ref{fig:rotation traps}.
In this implementation, the first and the last qubit in the linear cluster state exhibit lower error rates than the qubits in between.
The error in the first qubit is smaller due to the $\mathsf{iSWAP}$ error detection in the initialisation round.
The last qubit is not transferred back from the memory qubit to the network qubit for the final measurement, which reduces the number of $\mathsf{iSWAP}$ gates impacting this qubit by one.

\section{Error detection}
Errors during the initialisation step are detected in real time (error probability \ish\num{0.1}) in which case the current round is aborted and restarted after a period of cooling~\cite{drmota_robust_2023}.
The overhead due to repetition of this step is independent of the size of the cluster state and hence does not affect the scalability of the implemented protocol.
Here, the client can keep the same random variable $\tilde{\theta}_1$ because the phase shift $a_1 \pi$ depends on the photon polarisation measurement outcome, which is inherently random.
The server therefore obtains a maximally mixed state on the network qubit.
Information leaking through classical communication as a result of reusing $\tilde{\theta}_1$ could also be eliminated, as the client is in control of all relevant signals.
The error detection probability is closely related to the $\mathsf{iSWAP}$ infidelity; by reducing this infidelity, error detection would become obsolete.

\section{Photon loss}
The average time taken to obtain a single-photon herald is $\approx\SI{100}{\micro\second}$; hence the probability for no herald to occur within the timeout period of \SI{1}{\milli\second} is $<\num{e-4}$.
If no photon is received within this period, the client instructs the server to repeat the current round starting from the initialisation step after a series of system checks.
The timeout allows the server to recover from rare events that could temporarily hinder the generation of single photons, such as ion loss, misalignments in the photon collection optics, or laser failures.
The information leakage due to repetition in these rounds is purely classical and could be eliminated by the client, who is in control of all relevant signals.
Alternatively, the client could introduce fresh randomness at every retry to avoid this issue.
The loss of determinism due to photon loss is influenced by the timeout period and could be completely eliminated if the client continued execution of the protocol despite an unsuccessful steering attempt.
In this case, the server would check for system failures after the client has completed this round of the protocol.
The computational error that is incurred by this approach due to photon loss would reduce exponentially with increasing timeout period.

\section{Blindness}
\label{sec:Blindness}
Here we consider experimental imperfections that could adversely affect the blindness of the protocol.
Blindness is characterised by how much information an adversarial server could learn about the client's photon measurement basis choices $\theta_\ell$ (3 bits of information per shot).
Table~\ref{tab:leakage} summarises the sources of information leakage that are quantified in this section.
\begin{table}[h]
    \centering\footnotesize
    \begin{tabular}{cccll}
        \toprule
        \multirow{2}*{Channel} & \multirow{2}*{Source} & \multirow{2}*{Method} & \multicolumn{2}{c}{Leakage / bits} \\
        \cmidrule{4-5}
        & & & Observed & Optimised \\
        \midrule
        \ldelim\{{3}{*}[classical] & measurement angles & enforced & 0 & 0 \\
        {} & heralding efficiency & inferred & 0.00006 & 0.00006 \\
        {} & heralding delay & inferred & 0.35 & 0.00007 \\
        \midrule
        \ldelim\{{3}{*}[quantum] & measurement basis & inferred & 0.035 & 0.0007 \\
        {} & imbalanced outcomes & inferred & 0.00029 & 0.00026 \\
        {} & \textemdash & measured & 0.031(4) & \textemdash \\
        \bottomrule
    \end{tabular}
    \caption{Sources of information leakage in a single interaction step.
        The optimised values assume matched heralding delay (from excitation of \Sr{} until electronic detection) and balanced polarisation measurement outcomes.
        We use quantum state tomography to quantify the information that the server could gain from measuring the network qubit and find good agreement with independent estimates inferred from known imperfections (measurement basis, imbalanced outcomes).
        The values are to be compared with the amount of information (3 bits) that specifies the steered state, $\ket{\theta_\ell}$.
    }
    \label{tab:leakage}
\end{table}
At every shot of the experiment, indexed by $\ell$, we assume that the server has unrestricted access to
\begin{itemize}
    \item the qubit measurement angle $\delta_\ell$ (classical signal),
    \item the phase reference set by the photon detection time (classical signal),
    \item the number of attempts until a single-photon herald (classical signal),
    \item and the state of the memory and network qubits.
\end{itemize}
Blindness is compromised if the above observables correlate with $\theta_\ell$.

\begin{figure}[htb]
    \centering
    \includegraphics{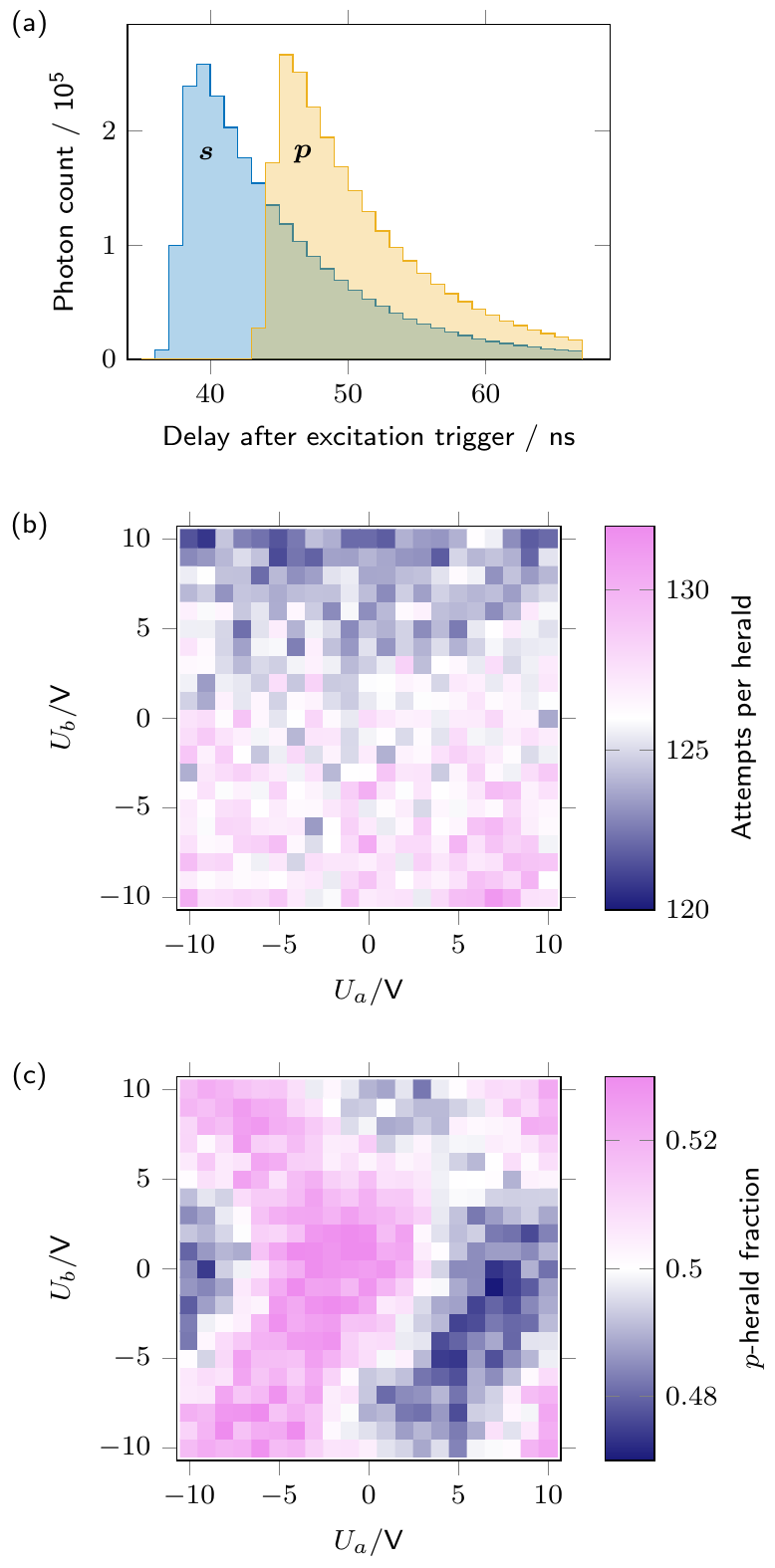}
    \caption{%
        (a) Histogram of delay between pulsed excitation of \Sr{} and a single-photon herald from detectors $s$ and $p$.
        (b) Average number of attempts (\SI{1}{\micro\second} per attempt) until a single-photon herald is announced by the client as a function of \ac{EOM} control voltages.
        (c) Ratio of heralds with outcome $p$ with respect to all heralds as a function of \ac{EOM} control voltages.
    }
    \label{fig:blindness}
\end{figure}

\subsection{Classical information leaks}
\label{sec:ClassicalLeakage}
In this section, we analyse the effect of classical sources of information leakage that are present in our system.
These sources do not include side-channel attacks, which are in general difficult to treat systematically.
We note that, in a real deployment, the client could straightforwardly monitor and eliminate leakage on all classical channels.

\subsubsection{Encrypted measurement angle}
$\delta_\ell$ does not leak information because it is encrypted with private randomness at every shot by the client.
We note that the use of a pseudo-random number generator does not compromise blindness (against computationally-bounded adversaries) provided that its implementation is cryptographically secure.

\subsubsection{Single-photon heralding delays}
The arrival time of photons at the \acp{APD} with respect to the pulsed excitation of \Sr{}, is inherently random due to the nature of spontaneous decay.
As these heralding times set the phase reference for the network qubit, they must be communicated to the server.
Statistical differences between the herald timing distributions could therefore be exploited by the server to learn about the polarisation measurement outcome.

The server could employ a maximum-likelihood strategy to guess the photon measurement result.
The detector properties may be known to the server.
This includes the probability distributions $\mathrm{Pr}_{s}(x)$ and $\mathrm{Pr}_{p}(x)$ for a herald at time $x$ in detector $s$ and detector $p$, respectively.
The server would assume that a photon heralded at time $x$ was observed in detector $s$ if $\mathrm{Pr}_s(x) > \mathrm{Pr}_p(x)$, and vice versa.
This maximum-likelihood strategy succeeds with probability $\sigma(x)$.
The expected information gain is given by $\mathcal{R}_\mathrm{ML} = \sum_x \mathrm{Pr}(x) \mathcal{R}(\sigma(x))$, where $\mathcal{R}(q) = 1 + q \log_2(q) + (1-q) \log_2(1-q)$ is the mutual entropy for a binary guess with success probability $q$, and $\mathrm{Pr}(x)$ is the overall probability of observing the arrival time $x$.

In this demonstration, no measures were taken to equalise the detector responses of the \acp{APD}.
Observed differences [Fig.~\ref{fig:blindness}(a)] are mainly due to unmatched delays in the electronic signal chain, by \SI{6.57(2)}{\nano\second}, leading to $\mathcal{R}_\mathrm{ML} \approx 0.35$ bits information leakage.
This imperfection could be minimised by adding delay to one of the signals (e.g.\ by extending the cable by $\approx\SI{1.3}{\meter}$).
After correcting for the delay mismatch, the information content would be reduced to $\mathcal{R}_\mathrm{ML} \approx \num{7e-5}$ bits, dominated by higher-order electronic distortion, such as differences in timing jitter between the electronic inputs used for timestamping.
Here, a field-programmable gate array is used for timestamping; it only provides \SI{1}{\nano\second} resolution, with inconsistent jitter between inputs.
We note that this is not a fundamental limitation and leakage of this kind could be eliminated by the client by appropriately conditioning the classical heralding signals.

\subsubsection{Single-photon heralding probability}
As the client publicly announces when a photon was successfully heralded, the server learns about the number of attempts taken, hence the detection efficiency, of the client apparatus.
If the detection efficiency depends on the secret measurement settings chosen by the client, the server could, in principle, obtain information that would compromise the blindness of the protocol.
Here we analyse the extent to which this affects our demonstration of blind quantum computing.

We observe variations in the detection efficiency depending on the \ac{EOM} voltage settings of $\approx\pm\SI{5}{\percent}$ [Fig.~\ref{fig:blindness}(b)].
The \ac{EOM} voltages are changed after every photon received by the client.\footnote{
    In \ltish\SI{5}{\percent} of cases in the present implementation, retries (due to errors in the initialisation step, with probability \ish\num{0.1}, and failures to produce a single-photon herald within the time-out interval, with probability $<\num{e-4}$) were performed without the client changing the measurement settings.
    However, this additional information does not change the conclusion of this analysis and the issue could be eliminated by introducing fresh randomness at every retry as well.}
Therefore, the server obtains only one attempt-number sample per measurement setting.
At every shot, the server's \emph{a priori} knowledge is reset, due to this shot-by-shot randomisation, to the exponential distribution with an expectation value $\lambda_0$ estimated from the number of attempts in previous shots.
The Fisher information, $\mathcal{I}(\lambda) = \lambda^{-2}$, quantifies the information each sample contributes towards an updated estimate of the underlying distribution, i.e.\ the exponential distribution for this shot with secret-dependent expectation value $\lambda$.
This quantity needs to be compared to the relative entropy, $\mathcal{S}(\lambda_0, \lambda)$, between the exponential distributions with respective expectation values $\lambda$ and $\lambda_0$.

For the observed average number of attempts, $\lambda_0 = 126$, and the maximum deviation $\Delta\lambda = |\lambda - \lambda_0| = 6$, the information gained per shot, $\mathcal{I}(\lambda) \ltish \num{e-4}$, is negligible.
Even if unlimited attempt-number samples were available for a given measurement setting, the information gain by the server could never exceed $\mathcal{S}(\lambda_0, \lambda)\,\ish \num{e-1}$ bits.

\subsection{Quantum information leaks}
We analyse the information which the server could gain on average from measuring the network qubit after \ac{RSP} in terms of the Holevo bound~$\chi$.
For perfect blindness, $\chi=0$, whereas one qubit can maximally transmit $\chi=1$ bit.
We exclude test rounds, where the client steers the network qubit into $\mathsf{Z}$ basis eigenstates, from this analysis.
Including dummy qubits would strictly reduce the Holevo information because these states do not correlate with $\theta_\ell$.
Hence, the resulting improvement in overall privacy depends on the fraction of rounds which are tests.

The server receives quantum states that are generated by steering of the network qubit.
For every qubit in the computation, the client chooses secretly one of 4 measurement bases, $\hat{B}_i$, $i\in\{1,2,3,4\}$, to steer the network qubit into $\rho_i^{+}$ with probability $q_i$ or $\rho_i^{-}$ with probability $(1-q_i)$, depending on the measurement outcome.
Averaged over the outcomes\footnote{
    If the basis angle $\theta_\ell$ is leaked, a simple measurement of the network qubit deterministically reveals also the client's measurement outcome.
}, the server receives the states $\rho_i = q_i \rho_i^{+} + (1-q_i) \rho_i^{-}$ with equal frequency:
\begin{align*}
    \rho = \frac{1}{4} \sum_{i=1}^{4} \rho_i \ .
\end{align*}
The Holevo information, which bounds the amount of information contained in this quantum state, is given by
\begin{align}
    \chi(\rho) = S(\rho) - \frac14 \sum_{i=1}^{4} S(\rho_i) \ ,
    \label{eq:Holevo}
\end{align}
where $S(\rho) = -\mathrm{Tr} \left( \rho \log_2(\rho) \right)$.

In our experiment, $\rho_i^{+}$ and $\rho_i^{-}$ are steered by the measurement of an entangled subsystem (a photon), which can be described using \acp{POVM}.
We use the maximally-entangled Bell state $(\ket{s}\ket{0} + \ket{p}\ket{1})/\sqrt{2}$ between the photon and the environment to simulate \ac{RSP}, which is justified by the proof in the section \enquote{\acl{RSP} by steering}.

\subsubsection{Imbalanced and mixed outcomes}
We estimate the blindness of \ac{RSP} using the \acp{POVM} $\{\hat{F}\}$ describing a \ac{PBS} with extinction ratios $\epsilon_{s}$ and $\epsilon_{p}$ for $s$- and $p$-polarised input, and two \acp{APD} that are placed in the (ideally) $s$- and $p$-polarised output ports of the \ac{PBS}, characterised by detection efficiencies $\eta_{s}$ and $\eta_{p}$, respectively:
\begin{align*}
    \hat{F}_p &= \eta_p (1-\epsilon_s)\ketbra{p}{p} + \eta_s \epsilon_p \ketbra{s}{s} \ , \\
    \hat{F}_s &= \eta_s (1-\epsilon_p)\ketbra{s}{s} + \eta_p \epsilon_s \ketbra{p}{p} \ , \\
    \hat{F}_{0} &= \mathsf{1} - \hat{F}_p - \hat{F}_s \ ,
\end{align*}
where $\hat{F}_p$ and $\hat{F}_s$ correspond to photon detection in \ac{APD} $p$ and \ac{APD} $s$, respectively, and $\hat{F}_0$ corresponds to a photon loss event.
As long as \ac{PBS} imperfections act as depolarising noise, such as for finite polarisation extinction, these effects have no effect on the Holevo information.
We may therefore assume without loss of generality $\forall i$ that $\rho_i^{+}$ and $\rho_i^{-}$ are orthogonal pure states, which occur with probability $q_i = \frac{\eta_p}{\eta_p+\eta_s} \approx \frac12$.
Choosing the more favourable configuration of measurement bases\footnote{
    It is less favourable to let $\rho_j^{+} = \ketbra{j \pi /4}{j\pi /4}$, as the error due to imbalanced rates adds constructively.
    However, using $\{\rho_1^{+}, \rho_2^{-}, \rho_3^{+}, \rho_4^{+}\}_j = \ketbra{j \pi /4}{j\pi /4}$ for example, part of the error cancels.
}, the average information contained in one qubit is bounded from above by
\begin{align}
    \chi = \frac{(6-\sqrt{2})\left(q-\frac12\right)^2}{4 \log_2(2)} + \mathcal{O}\left(\left(q-\frac12\right)^4\right) \ .
    \label{eq:Holevo-imbalance}
\end{align}
With the average imbalance $\langle |0.5 - q| \rangle_{U_a, U_b}$ derived from measurements shown in Fig.~\ref{fig:blindness}(c) and specified in the main text, we infer that the information leakage due to this effect is $\chi\approx\num{2.9e-4}$ bits.
Even if static photon loss is introduced to balance the detection efficiency on both detectors, the secret-dependent variation remains the dominant source of leakage: we estimate $\chi\approx\num{2.6e-4}$ bits for this scenario.
While the origin of the variations shown in Fig.~\ref{fig:blindness}(c) has not been confirmed, we suspect that the \acp{EOM} disturb the wavefront of the photons, which affects the mode matching into the optical fibres that are used for convenience to couple the photons into the \acp{APD}.
Because the quantum information is encoded in the polarisation of the photon, their spatial mode is irrelevant to our implementation, rendering this observation of voltage-dependent imbalance of heralds practically insignificant.
Alternatively, this source of leakage could be eliminated fully if the client conditioned heralds on local randomness (at the cost of a \SI{50}{\percent} reduction in the total heralding rate).

\subsubsection{Rotated POVMs}
Using the \acp{POVM} in the product space of the photon polarisation and the network qubit,
\begin{align*}
    \hat{G}_p &= \hat{F}_p \otimes \mathsf{1} \ , \\
    \hat{G}_s &= \hat{F}_s \otimes \hat{Z}(\phi) \ , \\
    \hat{G}_0 &= (\mathsf{1} \otimes \mathsf{1}) - \hat{G}_p - \hat{G}_s \ ,
\end{align*}
the Holevo information can be approximated by $\chi \approx 0.103 \times \phi^2$.

The electronic signal delay between orthogonal heralds shown in Fig.~\ref{fig:blindness}(a) not only leaks information via the classical reference time signal, but also causes a relative $\mathsf{Z}$ rotation of the network qubit depending on the photon measurement outcome.
Inserting $\phi = \Omega_Z \Delta t$, where $\Omega_Z \approx 2\pi \times \SI{14}{\mega\hertz}$ is the \Sr{} Zeeman qubit splitting and $\Delta t = \SI{6.57(2)}{\nano\second}$ is the measured delay difference between the two heralds [Fig.~\ref{fig:blindness}(a)], we obtain $\chi\approx 0.035$ bits.

Once the heralding delays are matched, imperfections in the photon polarisation measurement are expected to dominate information leakage through this channel.
% solve 0.0016 == \sin(\phi/2)^2 for \phi... = 0.08 rad.
Using the average overlap of $0.0016$ between projectors corresponding to different heralds [Fig.~\ref{fig:Orthogonality}], we obtain $\phi \approx \SI{0.08}{\radian}$, resulting in $\chi\approx\num{7e-4}$ bits of potential information leakage.

Alternatively, rather than relying on the orthogonality of the polarisation measurement, the client could exploit the fact that every polarisation measurement basis can be reached by two distinct \ac{EOM} voltage settings and switch randomly between them.
The effectiveness of this approach would depend on the accuracy of the voltage calibration, which could be increased by acquiring calibration data with sufficient precision.

\subsection{Measured quantum information leaks}
We use quantum state tomography to reconstruct the steered state of the network qubit for all eight equatorial target states of the measurement-based protocol (see average fidelities w.r.t.\ $\ket{\theta}$ in Fig.~\ref{fig:stability}(c)).
Using the estimated density matrices, we calculate the Holevo information according to Eq.~\eqref{eq:Holevo} and obtain of $0.031(4)$ bits per qubit, averaged over all runs where the polarisation measurement bases matched the (more favourable) configuration that was assumed above in the derivation of Eq.~\eqref{eq:Holevo-imbalance}.

\section{Security and Robustness for 2-qubit linear cluster states}
\begin{figure}[htbp]
    \centering
    \includegraphics{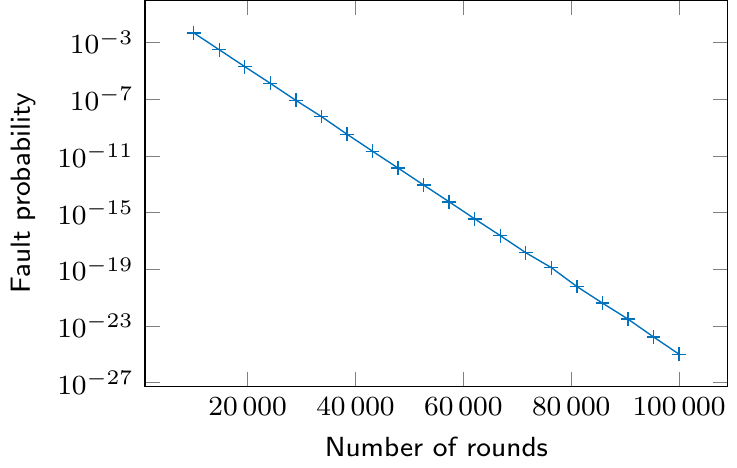}
    \caption{%
        We evaluate $p_\mathrm{fault}$ for a 2-qubit cluster state as a function of the number of rounds in the protocol proposed in Ref.~\cite{leichtle_verifying_2021}, assuming $p_\mathrm{max} = 0.185$.
        We obtain $\tau \eqsim 0.6$ and $\omega_\mathrm{max} \eqsim 0.205$ for all data points shown (crosses).
    }
    \label{fig:fault-probability}
\end{figure}
We evaluate an upper bound to the probability of the client accepting an incorrect result, $\mathrm{Pr}[\mathrm{fail}]$, for the verification protocol proposed in Ref.~\cite{leichtle_verifying_2021} by numerically minimising Eq.~(E5) from Ref.~\cite{leichtle_verifying_2021} for $k=2$ qubits and an intrinsic algorithmic error probability $p=0$.
During this minimisation, we constrain the probability of the client rejecting any result due to noise, $p_\mathrm{rej}$, by setting $p_\mathrm{rej} = \mathrm{Pr}[\mathrm{fail}]$.
The minimum fault probability, $p_\mathrm{fault} := \min_\mathcal{C} \mathrm{Pr}[\mathrm{fail}]$, where $\mathcal{C}$ denotes the constraints, is shown in Fig.~\ref{fig:fault-probability} for a variable number of total rounds, $n$, including $\tau n$ test rounds.
The associated optimal threshold, $\omega$, the fraction of test rounds, $\tau$, and security parameters ($\varphi, \epsilon_1, \epsilon_2, \epsilon_3$) are adjusted by the minimiser.
The probability of a failure, i.e., the probability of a security issue (accepting an incorrect result, $\mathrm{Pr}[\mathrm{fail}]$) and that of a robustness issue (rejecting any result, $p_\mathrm{rej}$), decreases exponentially with the number of rounds.

\section{Remote state preparation by steering}\label{sec:steering-proof}

In this section, we show that steering can be used to securely implement remote state preparation.
In particular, only one-way quantum communication from the server to the client is sufficient, and the preparation of entangled pairs by the server does not need to be trusted.

\begin{figure}[ht]
    \begin{resource}[H]
        \caption{Remote State Preparation}
        \label{resource:rsp}
        \vspace*{0.2cm}
        \begin{algorithmic}[0]
            \STATE \textbf{Inputs:}
            \begin{itemize}
                \item Client: the classical description of a single-qubit unitary $U$.
                \item Server: no input.
            \end{itemize}
            \STATE \textbf{Outputs:}
            \begin{itemize}
                \item Client: no output.
                \item Server: the single-qubit state $U|0\rangle$.
            \end{itemize}
        \end{algorithmic}
    \end{resource}
\end{figure}

\begin{figure}[ht]
    \begin{protocol}[H]
        \caption{RSP by steering}
        \label{protocol:rsp-steering}
        \vspace*{0.2cm}
        \begin{algorithmic}[0]
            \STATE \textbf{Inputs:}
            \begin{itemize}
                \item Client: the classical description of a single-qubit unitary $U$.
                \item Server: no input.
            \end{itemize}
            \STATE \textbf{Required resources:}
            \begin{itemize}
                \item Secure one-way quantum channel from server to client.
                \item Secure one-way classical channel from client to server.
            \end{itemize}
            \STATE \textbf{Instructions:}
            \begin{enumerate}
                \item The server prepares a two-qubit Bell state $|\Psi\rangle = \frac{1}{\sqrt{2}} \left( |00\rangle + |11\rangle \right)$, and sends one of the qubits to the client.
                \item The client samples a single-qubit unitary $U_1$ randomly from the Haar measure.
                It then applies $U_1$ to the state received by the server and performs a measurement on it in the computational basis, obtaining measurement outcome $m$.
                \item The client sends the classical description of the single-qubit unitary $U_2 = U \mathrm{X}^m U_1$ to the server.
                \item The server applies $U_2$ to the remaining single-qubit state, and sets it as its output.
            \end{enumerate}
        \end{algorithmic}
    \end{protocol}
\end{figure}

\begin{lemma}
    Protocol~\ref{protocol:rsp-steering} implements Resource~\ref{resource:rsp} with perfect correctness.
\end{lemma}
\begin{proof}
    After the client's measurement, the remaining single-qubit state in the server's quantum register can be described by $U_1^\dagger \mathrm{X}^m |0\rangle$. The server's output therefore becomes $U_2 U_1^\dagger \mathrm{X}^m |0\rangle = U |0\rangle$.
\end{proof}

\begin{lemma}
    Protocol~\ref{protocol:rsp-steering} implements Resource~\ref{resource:rsp} with perfect security against a malicious server.
\end{lemma}
\begin{proof}
    As part of the proof, let $\sigma$ be defined as in Simulator~\ref{simulator:rsp-steering-security}.
    \begin{figure}[hb]
    \begin{simulator}
        \caption{}
        \label{simulator:rsp-steering-security}
        \vspace*{0.2cm}
        \begin{algorithmic}[0]
            \STATE \textbf{Instructions:}
            \begin{enumerate}
                \item The simulator expects a single-qubit quantum state $|\phi_1\rangle$ as an input from the ideal functionality on its left interface.
                \item It expects a single-qubit quantum state $|\phi_2\rangle$ as an input from the distinguisher on its right interface.
                \item It samples a single-qubit unitary $U_1$ randomly from the Haar measure.
                \item It applies the two-qubit unitary $U_1 \otimes \mathrm{I}$ to the state $|\phi_1\rangle|\phi_2\rangle$ and performs a Bell measurement on it, obtaining measurement outcomes $m_1$ and $m_2$.
                \item It then sets $U_2 = U_1^\dagger \mathrm{Z}^{m_1} \mathrm{X}^{m_2}$ and outputs the classical description of $U_2$ on its right interface to the distinguisher.
            \end{enumerate}
        \end{algorithmic}
    \end{simulator}
    \end{figure}

    It remains to be shown that the composition of Resource~\ref{resource:rsp} with $\sigma$ (the ideal world) generates the same distribution on its interfaces as the client's instructions of Protocol~\ref{protocol:rsp-steering} (the real world).

    Let $|\psi\rangle$ be the purification of the distinguisher's quantum register just before sending the first qubit of its register to the client. Then, in the real world, after the client's measurement, the state of the server's quantum register can be described (up to a global phase) by
    \begin{align*}
        \langle 0| ( \mathrm{X}^m U_1 \otimes \mathrm{I}) |\psi\rangle,
    \end{align*}
    and the classical message from the client contains the description of the unitary $U X^m U_1$, where $U_1$ is chosen according to the Haar measure. Substituting $U_1$ for $\mathrm{X}^m U_1$ yields
    \begin{align*}
        \langle 0| ( U_1 \otimes \mathrm{I}) |\psi\rangle,
    \end{align*}
    and the classical description of $U U_1$ without changing the distribution of $U_1$.

    In the ideal world, after the simulator's measurement, the state of the server's quantum register can be described (up to a global phase) by
    \begin{align*}
        \langle \Phi^+| ( \mathrm{Z}^{m_1} \mathrm{X}^{m_2} U_1 \otimes \mathrm{I} ) ( U|0\rangle \otimes |\psi\rangle ),
    \end{align*}
    and the classical message from the client contains the description of the unitary $U_1^\dagger \mathrm{Z}^{m_1} \mathrm{X}^{m_2}$, where $U_1$ is chosen according to the Haar measure. This can be rewritten equivalently as
    \begin{align*}
        \langle 0| ( U^\dagger U_1^\dagger \mathrm{Z}^{m_1} \mathrm{X}^{m_2} \otimes \mathrm{I} ) |\psi\rangle.
    \end{align*}
    A change of variables from $U_1$ to $\mathrm{Z}^{m_1} \mathrm{X}^{m_2} U_1^\dagger U^\dagger$, without changing the distribution of $U_1$, yields
    \begin{align*}
        \langle 0| ( U_1 \otimes \mathrm{I} ) |\psi\rangle,
    \end{align*}
    and the classical description of $(\mathrm{Z}^{m_1} \mathrm{X}^{m_2} U_1^\dagger U^\dagger)^\dagger \mathrm{Z}^{m_1} \mathrm{X}^{m_2} = U U_1$, which concludes the proof.
\end{proof}

\begin{acronym}
    \acro{BQP}{bounded-error quantum polynomial time}
    \acro{BQC}{blind quantum computing}
    \acro{RSP}{remote state preparation}
    \acro{EOM}{electro-optic modulator}
    \acro{PBS}{polarising beamsplitter}
    \acro{APD}{avalanche photodiode}
    \acro{POVM}{positive operator-valued measurement}
\end{acronym}

\end{document}